\RequirePackage[l2tabu,orthodox]{nag}
\RequirePackage{amsmath,amssymb}
\documentclass[a4paper]{scrartcl}
\usepackage{fullpage}

\usepackage{tikz}
\usetikzlibrary{decorations.pathreplacing,angles,quotes}
\usepackage{standalone}

\usepackage[utf8]{inputenc}
\usepackage[T1]{fontenc}
\usepackage{lmodern} 

\usepackage{microtype}
\usepackage[within=none,figurename=Fig.,labelsep=period]{caption}

\usepackage{amssymb}
\usepackage{amsmath}
\usepackage{bm}   
\usepackage{mathtools}

\allowdisplaybreaks[1]

\usepackage[inline]{enumitem}
\setlist[itemize]{label=$\circ$}
\setlist[description]{labelindent=\parindent}

\usepackage[amsmath,hyperref,thmmarks]{ntheorem}

\theoremseparator{.}
\newtheorem{theorem}{Theorem}
\newtheorem{lemma}[theorem]{Lemma}

\newtheorem{fact}[theorem]{Fact}

\newtheorem{corollary}[theorem]{Corollary}
\newtheorem{claim}[theorem]{Claim}

\theoremstyle{plain}
\theorembodyfont{\upshape}

\theoremstyle{nonumberplain}
\theoremheaderfont{\itshape}
\theorembodyfont{\upshape}
\theoremsymbol{\ensuremath{_\blacksquare}}
\newtheorem{proof}{Proof}

\newcommand{\cc}[1]{\ensuremath{\mathrm{#1}}}

\newcommand{\W}{\cc{W[1]}}

\newcommand{\parared}{\leq^T_P}

\newcommand{\homs}{\ensuremath{\mathsf{Hom}}}
\newcommand{\embs}{\ensuremath{\mathsf{Emb}}}
\newcommand{\subs}{\ensuremath{\mathsf{Sub}}}
\newcommand{\strembs}{\ensuremath{\mathsf{StrEmb}}}
\newcommand{\indsubs}{\ensuremath{\mathsf{IndSub}}}
\newcommand{\auts}{\ensuremath{\mathsf{Aut}}}

\newcommand{\modqa}{\mathsf{Mod}}
\newcommand{\isof}{\mathsf{Iso}}
\newcommand{\epk}{\mathsf{E}^\Phi_k}

\newcommand{\dense}{dense}
\newcommand{\Symgp}{\ensuremath{\mathsf{Sym}}}

\title{Counting Induced Subgraphs: A Topological Approach to \#W[1]-hardness}

\usepackage{authblk}

\author[a]{Marc Roth}
\author[b]{Johannes Schmitt}

\affil[a]{%
  Saarland University and Cluster of Excellence (MMCI), Saarbrücken, Germany
}

\affil[b]{%
  ETH Z\"urich, Zürich, Switzerland
}

\begin{document}
 
\maketitle

\begin{abstract}
We investigate the problem $\#\indsubs(\Phi)$ of counting all induced subgraphs of size~$k$ in a graph~$G$ that satisfy a given property $\Phi$. This continues the work of Jerrum and Meeks who proved the problem to be $\#\W$-hard for some families of properties which include, among others, (dis)connectedness~[JCSS~15] and even- or oddness of the number of edges~[Combinatorica~17]. Using the recent framework of graph motif parameters due to Curticapean, Dell and Marx [STOC~17], we discover that for monotone properties~$\Phi$, the problem $\#\indsubs(\Phi)$ is hard for $\#\W$ if the reduced Euler characteristic of the associated simplicial (graph) complex of $\Phi$ is non-zero. This observation links $\#\indsubs(\Phi)$ to Karp's famous Evasiveness Conjecture, as every graph complex with non-vanishing reduced Euler characteristic is known to be evasive. Applying tools from the ``topological approach to evasiveness'' which was introduced in the seminal paper of Khan, Saks and Sturtevant [FOCS~83], we prove that $\#\indsubs(\Phi)$ is~$\#\W$-hard for every monotone property $\Phi$ that does not hold on the Hamilton cycle as well as for some monotone properties that hold on the Hamilton cycle such as being triangle-free or not $k$-edge-connected for $k > 2$. Moreover, we show that for those properties $\#\indsubs(\Phi)$ can not be solved in time $f(k)\cdot n^{o(k)}$ for any computable function $f$ unless the Exponential Time Hypothesis (ETH) fails. In the final part of the paper, we investigate non-monotone properties and prove that $\#\indsubs(\Phi)$ is $\#\W$-hard if $\Phi$ is any non-trivial modularity constraint on the number of edges with respect to some prime $q$ or if $\Phi$ enforces the presence of a fixed isolated subgraph.
\end{abstract}

\section{Introduction}\label{sec:intro}
In their work about the parameterized complexity of counting problems~\cite{flumgrohe_counting} Flum and Grohe introduced the parameterized analogue of the theory of computational counting as layed out by Valiant in his seminal paper about the complexity of computing the permanent~\cite{valiant}. Since then parameterized counting has evolved into a well-studied subfield of parameterized complexity theory. In particular, there has been remarkable progress in the classification of problems that require to count small structures in large graphs. It turned out that many families of such counting problems allow so-called dichotomy results, that is, every problem in the family is either fixed-parameter tractable or hard for the class $\#\W$ --- the counting equivalent of $\W$. One result of that kind is the dichotomy for counting homomorphisms \cite{homsdicho1,homsdicho2}. Here one is given a graph $H$ from a class of graphs $\mathcal{H}$ and an arbitrary graph $G$ and the task is to compute the number of homomorphisms from $H$ to $G$. When parameterized by $|H|$ this problem is fixed-parameter tractable if there exists a constant upper bound on the treewidth of graphs in $\mathcal{H}$ and $\#\W$-hard otherwise. Similar results have been shown for the problems of counting subgraph embeddings \cite{embsdicho}, induced subgraphs \cite{strembsdicho} and locally injective homomorphisms~\cite{rothmobius}. As results like Ladner's theorem (see e.g. \cite{classic_ladner,para_ladner}) rule out such dichotomies in the general case one might ask why all of the above problems indeed do allow such complexity classifications. The answer to that question was given very recently by Curticapean, Dell and Marx~\cite{hombasis2017} who proved that, in some sense, all of those problems are the same. To this end, they defined the problem of computing linear combinations of homomorphisms which they called \emph{graph motif parameters}. Here one is given a graph $G$ and a function $a$ of finite support that maps graphs to rational numbers and the task is to compute
\begin{equation}
\sum_H a(H) \cdot \#\homs(H,G) \,,
\end{equation}
where the sum is over all (unlabeled) simple graphs and $\#\homs(H,G)$ denotes the number of homomorphisms from $H$ to $G$. A result of Lov{\'a}sz (see e.g. Chapt.~5 in~\cite{lovaszbook}) implies that the number of subgraph embeddings $\#\embs(H,G)$ as well as the number of induced subgraphs $\#\indsubs(H,G)$ can be expressed as a linear combination of homomorphisms. In case of embeddings the result states that
\begin{equation}
\#\embs(H,G) = \sum_{\rho \geq \emptyset} \mu(\emptyset,\rho) \cdot \#\homs(H/\rho,G)\,,
\end{equation}
where the sum is over the partition lattice of the vertices of $H$, $\mu$ is the Möbius function over that lattice and $H/\rho$ is obtained from $H$ by identifying vertices along $\rho$. Now, intuitively, the main result of Curticapean, Dell and Marx states that computing a linear combination of homomorphisms is precisely as hard as computing the hardest term in the linear combination. Together with the dichotomy for counting homomorphisms this implies that every problem expressible as a linear combination of homomorphisms is either fixed-parameter tractable or $\#\W$-hard. 

The purpose of this work is a thorough investigation of the problem of counting induced subgraphs through the lense of the framework of graph motif parameters. Chen, Thurley and Weyer \cite{strembsdicho} proved that the problem $\#\indsubs(\mathcal{H})$ of, given a graph $H \in \mathcal{H}$ and an arbitrary graph $G$, computing $\#\indsubs(H,G)$ is fixed-parameter tractable when parameterized by $|H|$ if and only if $\mathcal{H}$ is finite and $\#\W$-hard otherwise. While this result resolves the parameterized complexity of problems such as computing the number of induced cycles of length $k$,\footnote{This problem can be equivalently expressed as $\#\indsubs(\mathcal{C})$, where $C$ is the class of all cycles.} it is not applicable to problems such as computing the number of connected induced subgraphs of size $k$. For this reason, Jerrum and Meeks \cite{indsubs_connected,indsubs_density,indsubs_treewidth,indsubs_eveness} introduced and studied the following problem: Let $\Phi$ be a (computable) graph property, then the problem $\#\indsubs(\Phi)$ asks, given a graph $G$ and a natural number $k$, to count all induced subgraphs of size $k$ in $G$ that satisfy $\Phi$. In other words, the goal is to compute
\begin{equation}
\sum_{H \in \Phi_k} \#\indsubs(H,G) \,,
\end{equation}  
where $\Phi_k$ is the set of all (unlabeled) graphs with $k$ vertices that satisfy $\Phi$. The generality of $\#\indsubs(\Phi)$ allows to count almost arbitrary substructures in graphs, subsuming lots of parameterized counting problems that have been studied before, and hence the problem deserves a thorough complexity analysis with respect to the property $\Phi$. Jerrum and Meeks proved it to be $\#\W$-hard for the property of connectivity~\cite{indsubs_connected}, for the property of having an even (or odd) number of edges~\cite{indsubs_eveness} as well as for some other properties (see Section~\ref{sec:related_work}). As noted in \cite{hombasis2017}, the theory of graph motif parameters immediately implies that for every property $\Phi$, the problem $\#\indsubs(\Phi)$ is either fixed-parameter tractable or $\#\W$-hard. However, for a concrete $\Phi$ it might be highly non-trivial to prove for which graphs $H$ the term $\#\homs(H,G)$ is contained with a non-zero coefficient in the equivalent expression as linear combination of homomorphisms. Unfortunately, this is precisely what needs to be done to find out whether $\#\indsubs(\Phi)$ is fixed-parameter tractable or $\#\W$-hard. In our investigation we will focus on the coefficient of $\#\homs(K_k,G)$, where $K_k$ is the complete graph on k vertices. We will see that for monotone properties, non-zeroness of this coefficient is sufficient for the property to be evasive.

\subsection{Results and techniques}
The framework of graph motif parameters~\cite{hombasis2017} implies that for every property $\Phi$ and natural number $k$, there exists a function $a$ from graphs to rationals such that for all graphs $G$ it holds that
\begin{equation}
\sum_{H \in \Phi_k} \#\indsubs(H,G) = \sum_H a(H) \cdot \#\homs(H,G) \,.
\end{equation}
Our most important observation is concerned with the coefficient of the complete graph.
\begin{theorem}[Intuitive version] \label{thm:clique_coef_int}
Let $\Phi$, $k$ and $a$ be as above. Then it holds that $a(K_k)= 0$ if and only if $\sum_{A \in \epk}  (-1)^{\#A} = 0$, where $\epk$ is the set of all edge-subsets $A$ of the labeled complete graph with $k$ vertices such that $\Phi$ holds for the graph induced by $A$. 
\end{theorem}
It turns out that for monotone properties, i.e., properties that are closed under the removal of edges, the term $\sum_{A \in \epk} (-1)^{\#A}$ is equal to the reduced Euler characteristic $\hat{\chi}$ of the simplicial graph complex of $\Phi_k$. Recall that a simplicial complex is a set of sets that is closed under taking non-empty subsets and a simplicial graph complex is a simplicial complex whose elements are subsets of the edges of the labeled complete graph. We will make this formal in Section~\ref{sec:prelims}. Applying Theorem~1 to monotone properties we hence obtain the following.

\begin{corollary}[Intuitive version] \label{cor:clique_coef_int_mono}
Let $\Phi$, $k$ and $a$ be as above and assume furthermore that $\Phi$ is monotone. Then $|k! \cdot a(K_k)| = \left| \hat{\chi}(\Delta(\Phi_k))\right|$ 
where $\Delta(\Phi_k)$ is the associated simplical graph complex of $\Phi_k$.
\end{corollary}

As computing the number of cliques of size $k$ is $\#\W$-complete~\cite{flumgrohe_counting} and computing a linear combination of homomorphisms is precisely as hard as computing its hardest term~\cite{hombasis2017}, Corollary~\ref{cor:clique_coef_int_mono} immediately resolves the complexity of $\#\indsubs(\Phi)$ whenever $\Phi$ is monotone and the reduced Euler characteristic of $\Delta(\Phi_k)$ is known to be non-zero for infinitely many $k$. Moreover, as the reduction in~\cite{hombasis2017} is tight, we also obtain a matching lower bound assuming the Exponential Time Hypothesis (ETH) if the set of such $k$ is dense. Here an infinite set $\mathcal{K}$ of natural numbers is \emph{\dense} if there exists a constant $c>0$ such that for all but finitely many $k \in \mathbb{N}$ there exists $k' \in \mathcal{K}$ such that $k\leq k' \leq c\cdot k$.

\begin{corollary}[Intuitive version] \label{cor:eul_non_zero_hard_int}
Let $\Phi$ be a monotone graph property such that $\hat{\chi}(\Delta(\Phi_k)) \neq 0$ for infinitely many $k$. Then the problem $\#\indsubs(\Phi)$ is $\#\W$-hard. If additionally the set of all $k$ such that $\hat{\chi}(\Delta(\Phi_k)) \neq 0$ is dense, it can not be solved in time $f(k)\cdot n^{o(k)}$ for any computable function $f$ unless ETH fails.
\end{corollary}
The (reduced) Euler characteristic is well-understood for many graph complexes. For\linebreak example, Chapt.~10.5 in the book of Jonsson~\cite{jonsson} provides a large list of graph properties, each of whose reduced Euler characteristics are non-zero infinitely often. For those properties Corollary~\ref{cor:eul_non_zero_hard_int} is hence applicable.

The study of the (reduced) Euler characteristic is, among others, motivated by Karp's famous evasiveness conjecture, stating that every non-trivial monotone graph property is evasive. A property $\Phi_k$ on graphs with $k$ vertices is evasive if every decision-tree algorithm that branches on the presence or absence of edges of a given graph $G$ needs to perform $\binom{k}{2}$ branches in the worst case to correctly decide whether $\Phi_k$ holds on $G$. We refer the reader to Miller's survey~\cite{evasiveness_survey} for a detailed introduction. While the conjecture is still unresolved, there has been a major breakthrough due to Khan, Saks and Sturtevant~\cite{topological_evasiveness} who proved the conjecture to be true whenever $k$ is a prime power. Their paper ``A Topological Approach to Evasiveness'' was, as the name suggests, the first one to use topological tools such as fixed-point complexes under group operations to prove evasiveness of a given graph complex. One of their results reads as follows\footnote{In fact, Khan, Saks and Sturtevant show that any non-evasive complex is collapsible. However, every collapsible complex has a reduced Euler characteristic of zero (see e.g.~\cite{evasiveness_lutz}). Hence the contraposition implies the theorem as stated.}.

\begin{theorem}[\cite{topological_evasiveness}] Let $\Phi_k$ be a non-trivial monotone graph property. If $\hat{\chi}(\Delta(\Phi_k)) \neq 0$ then $\Phi_k$ is evasive.
\end{theorem}

Unfortunately, the converse of this theorem does not hold. A counterexample is given in Chapt.~10.6 in Jonsson's book~\cite{jonsson}. Nevertheless it turns out that some tools of the topological approach to evasiveness suit as well for a topological approach to $\#\W$-hardness of $\#\indsubs(\Phi)$. The most important ingredient in our proofs is a theorem that goes back to Smith~\cite{smith} (see also~\cite{oliver} and Chapt.~3 in~\cite{bredon}), intuitively stating that, given a simplicial complex $\Delta$ and a $p$-power group $\Gamma$ for some prime $p$ that operates on the ground set of $\Delta$ in a way that leaves the complex stable, it holds that 
\begin{equation}
\hat{\chi}(\Delta) \equiv \hat{\chi}(\Delta^\Gamma) \mod p\,,
\end{equation}
where $\Delta^\Gamma$ is the fixed-point complex of $\Delta$ with respect to $\Gamma$. Again, this will be made formal in Section~\ref{sec:prelims}. Applying this theorem to a rather simple group, we will be able to prove our main result which reads as follows:
\begin{theorem}\label{thm:main}
Let $\Phi$ be a non-trivial monotone graph property. Then $\#\indsubs(\Phi)$ is $\#\W$-hard and, assuming ETH, can not be solved in time $f(k)\cdot n^{o(k)}$ for any computable function~$f$ if at least one of the following conditions is true
\begin{enumerate}
\item $\Phi$ is false for odd cycles.
\item $\Phi$ is true for odd anti-holes.
\item There exists $c \in \mathbb{N}$ such that for all $H$ it holds that $\Phi(H) = 1$ if and only if $H$ is not $c$-edge-connected.
\item There exists a graph $F$ such that for all $H$ it holds that $\Phi(H) = 1$ if and only if there is no homomorphism from $F$ to $H$.
\end{enumerate}
\end{theorem}

We remark that Rivest and Vuillemin~\cite{rivest} implicitly proved that the reduced Euler characteristic of a graph complex does not vanish if the first condition is true. Furthermore we note that (non-)triviality of a monotone property needs to be defined with some care to exclude properties that depend only on the number of vertices of a graph. Details are given in Section~\ref{sec:mono}. Examples of properties that satisfy the first condition are the ones of being bipartite, cycle-free, disconnected and non-hamiltonian. One example for the second condition is the property of having a chromatic number smaller or equal than half of the size of the graph (rounded up) and the third condition includes the properties of exclusion of a fixed complete graph as a subgraph.

Finally, we investigate $\#\indsubs(\Phi)$ for two families of non-monotone properties. For the first one, let $q$ be a prime and $\mathcal{Q}$ be a subset of $\{0,\dots,q-1\}$. Then the property $\modqa[q,\mathcal{Q}]$ holds on a graph $H$ if and only if $\left(\#E(H)~ \mathsf{mod} ~q\right) \in \mathcal{Q}$. For the second one, let $F$ be a connected graph. Then the property $\isof[F]$ holds on a graph $H$ if and only if $H$ contains an isolated subgraph that is isomorphic to $F$.

\begin{theorem}\label{thm:main_non_mono}
For all primes $q$, non-trivial subsets $\mathcal{Q}$ of $\{0,\dots,q-1\}$ and connected graphs~$F$, the problems $\#\indsubs(\modqa[q,\mathcal{Q}])$ and $\#\indsubs(\isof[F])$ are $\#\W$-hard and can not be solved in time $f(k)\cdot n^{o(k)}$ for any computable function $f$, unless ETH fails.
\end{theorem}

\subsection{Related work}\label{sec:related_work}
Jerrum and Meeks introduced and studied the problem $\#\indsubs(\Phi)$ for the following properties. In~\cite{indsubs_connected} they prove the problem to be $\#\W$-hard if $\Phi$ is the property of being connected, which immediately follows from Theorem~\ref{thm:main} as $\#\indsubs(\Phi)$ and $\#\indsubs(\neg\Phi)$ are equivalent\footnote{We just need to substract one from $\binom{n}{k}$ to get the other.} and the property of being disconnected is monotone and false for every cycle. In~\cite{indsubs_eveness} hardness is established for the property of having an even (or odd) number of edges, which is subsumed by Theorem~\ref{thm:main_non_mono}. Indeed, the case of $q=2$ follows already from Theorem~\ref{thm:clique_coef_int} as every term in the sum $\sum_{A \in \epk} (-1)^{\#A}$ will have the same sign. In~\cite{indsubs_density} Jerrum and Meeks prove the problem to be $\#\W$-hard whenever the edge-density of graphs in $\Phi_k$ grows asymptotically slower than $k^2$ and in~\cite{indsubs_treewidth} Meeks shows that whenever $\Phi$ is co-monotone, i.e., $\neg \Phi$ is monotone, and the set of (edge-)minimal elements of $\Phi$ has unbounded treewidth, the problem is hard as well. 

Those latter results are independent from ours in the sense that ours do not imply theirs and vice versa. One example of a property whose hardness does not follow from the results of Jerrum and Meeks is bipartiteness: The edge-densities of both, the properties of being bipartite and not bipartite grow asymptotically as fast as $k^2$ and the edge-minimal non-bipartite graphs are odd cycles, hence having treewidth $2$. However hardness for the property of being bipartite follows from the first condition of Theorem~\ref{thm:main} as odd cycles are not bipartite. Moreover, we remark that Meek's reduction in~\cite{indsubs_treewidth} uses the Excluded Grid Theorem and hence does not imply a tight lower bound under ETH. 

The remainder of the paper is structured as follows. In Section~\ref{sec:prelims} we introduce the necessary background in parameterized (counting) complexity, graph theory as well as in the theory of transformation groups and simplicial complexes. In Section~\ref{sec:hombasis} we give a formal introduction to graph motif parameters and prove Theorem~\ref{thm:clique_coef_int}. This is followed by the analysis of monotone properties and the proof of Theorem~\ref{thm:main} in Section~\ref{sec:mono}. In particular, we will take a close look at the fixed-point set of the group operation on labeled graphs that induces a cyclic shift on the vertices. Finally we discuss non-monotone properties and prove Theorem~\ref{thm:main_non_mono} in Section~\ref{sec:non_mono}.

\section{Preliminaries}\label{sec:prelims}
First we will introduce some basic notions. Given a finite set $S$, we write $\#S$ for the cardinality of $S$. We say that an infinite set $\mathcal{K}$ of natural numbers is \emph{\dense} if there exists a constant $c>0$ such that for all but finitely many $k \in \mathbb{N}$ there exists $k' \in \mathcal{K}$ such that $k\leq k' \leq c\cdot k$. Given a function $a$ from a (not necessarily finite) set $S$ to rational numbers, the \emph{support} of $a$ is the set of elements $s\in S$ such that $a(s) \neq 0$. We write $\mathsf{supp}(a)$ for the support of $a$. Given a natural number $k$, we write $[k]$ for the set $\{0,\dots,k-1\}$. Given a finite group $\Gamma$ of order $p^s$ for some prime $p$ and natural number $s$, we say that $\Gamma$ is a $p$-\emph{power group}. 

\paragraph*{Graph theory}
In this work all graphs are considered to be undirected, simple and to not contain self-loops. Given a graph $G$ we write $V(G)$ for the vertices and $E(G)$ for the edges of $G$. We denote the complete graph on $\ell$ vertices as $K_\ell$. A \emph{labeled} graph is a graph $G$ with a bijective labeling $\ell: V(G) \rightarrow \left[\#V(G)\right]$ of the vertices and we will sloppily identify vertices with their labels. A \emph{subgraph} of $G$ is a graph obtained from $G$ by deleting vertices (including incident edges) and/or edges. Given a subset $S \subseteq V(G)$, the \emph{induced subgraph} $G[S]$ is the graph with vertices $S$ and edges $E(G) \cap S^2$.

A \emph{homomorphism} from a graph $H$ to a graph $G$ is a function $\varphi: V(H) \rightarrow V(G)$ that is edge-preserving, i.e. for every edge $\{u,v\} \in E(H)$ it holds that $\{\varphi(u),\varphi(v) \} \in E(G)$. We write $\homs(H,G)$ for the set of all homomorphisms from $H$ to $G$. A homomorphism $\varphi$ is called an \emph{embedding} if $\varphi$ is injective. We write $\embs(H,G)$ for the set of all embeddings from $H$ to $G$. An \emph{isomorphism} from a graph $H$ to a graph $G$ is a bijective homomorphism. We say that $H$ and $G$ are \emph{isomorphic}, denoted by $H \cong G$, if such an isomorphism exists and we denote $\mathcal{G}$ as the set of all (isomorphism types of) graphs. An \emph{automorphism} of a graph $H$ is an isomorphism from $H$ to $H$. We write $\auts(H)$ for the set of all automorphisms of $H$. An embedding $\varphi$ from $H$ to $G$ is called a \emph{strong embedding} if for all vertices $u,v \in V(H)$ it holds that $\{u,v\} \in E(H) \Leftrightarrow \{\varphi(u),\varphi(v) \} \in E(G)$. We write $\strembs(H,G)$ for the set of all strong embeddings from $H$ to $G$. 

Given graphs $H$ and $G$, we write $\subs(H,G)$ for the set of all subgraphs of $G$ that are isomorphic to $H$ and $\indsubs(H,G)$ for the set of all induced subgraphs of $G$ that are isomorphic to $H$.

\begin{fact}\label{fact:auto}
For all graphs $H$ and $G$ it holds that $\#\embs(H,G) = \#\subs(H,G) \cdot \#\auts(H)$ and that $\#\strembs(H,G) = \#\indsubs(H,G) \cdot \#\auts(H)$.
\end{fact}

A \emph{graph property} $\Phi$ is a function from graphs to $\{0,1\}$ such that $\Phi(G) = \Phi(G')$ whenever $G$ and $G'$ are isomorphic. We say that $\Phi$ holds on $G$ if $\Phi(G) = 1$ and we are not going not distinguish between $\Phi$ and the set of graphs for which $\Phi$ holds as this will be clear from the context. We write $\Phi_k$ for the set of all (isomorphism types of) graphs with $k$ vertices on which $\Phi$ holds. For technical reasons we define $\epk$ to be the set of all edge-subsets $A$ of the labeled complete graph with $k$ vertices such that $\Phi$ holds on the graph with the same vertices and edges $A$. A graph property is called \emph{monotone} if it is closed under the removal of edges, that is, if $G'$ is obtained from $G$ by removing edges and $\Phi$ holds for $G$, then $\Phi$ holds for $G'$ as well. A property is called \emph{co-monotone} if its complement is monotone\footnote{We remark that in some literature, e.g. \cite{indsubs_treewidth,rivest}, the notions of monotonicity and co-monotonicity are reversed.}.

\paragraph*{Transformation groups and simplical (graph) complexes}
Let $\Omega$ be a finite set. A \emph{simplicial complex} over the ground set $\Omega$ is a set $\Delta$ of non-empty subsets of $\Omega$ such that whenever a set $A$ is contained in $\Delta$ and $A'$ is a non-empty subset of $A$, then $A'$ is contained in $\Delta$ as well. An element $A$ of $\Delta$ is called a \emph{simplex} and the \emph{dimension} of $A$, denoted as $\mathsf{dim}(A)$, is defined to be $\#A -1$. The \emph{Euler characteristic} $\chi$ of a simplical complex $\Delta$ is defined to be $\chi(\Delta) :=\sum_{i\geq 0} (-1)^i \cdot \#\{ A \in \Delta ~|~\mathsf{dim}(A)=i \}$ and the \emph{reduced Euler characteristic} of $\Delta$ is defined to be $ \hat{\chi}(\Delta) := 1-\chi(\Delta)$.

\begin{fact}\label{fact:euler_red_euler}
$\hat{\chi}(\Delta) = \sum_{i\geq 0} (-1)^i \cdot \#\{ A \in \Delta\cup\{ \emptyset\} ~|~\#A=i \} $.
\end{fact}

Given a simplicial complex $\Delta$ and a finite group $\Gamma$ that operates on the ground set $\Omega$ of $\Delta$, we say that $\Delta$ is a \emph{$\Gamma$-simplicial complex} if the induced action of $\Gamma$ on subsets of $\Omega$ preserves $\Delta$. More precisely, if $A \in \Delta$ and $g \in \Gamma$ then the set $g \triangleright A := \{g \triangleright a ~|~ a \in A\}$ is contained in $\Delta$ as well. If this is the case we can define the \emph{fixed-point complex} $\Delta^\Gamma$ as follows. Let $\mathcal{O}_1,\dots,\mathcal{O}_{k}$ be the orbits of $\Omega$ with respect to $\Gamma$. Then
\[ \Delta^\Gamma := \left\lbrace S \subseteq \{1,\dots,k\} ~\middle|~ S \neq \emptyset \wedge \bigcup_{i \in S} \mathcal{O}_i \in \Delta \right\rbrace \]

The following theorem, which is due to Smith~\cite{smith} (see also~\cite{oliver} and Chapt.~3 in~\cite{bredon}), will be of crucial importance in Section~\ref{sec:mono}.

\begin{theorem}\label{thm:smith}
Let $\Gamma$ a group of order $p^s$ for some prime $p$ and natural number $s$ and let $\Delta$ be a $\Gamma$-simplicial complex. Then $\chi(\Delta) \equiv \chi(\Delta^\Gamma) \mod p$ and hence $\hat{\chi}(\Delta) \equiv \hat{\chi}(\Delta^\Gamma) \mod p$.
\end{theorem} 

Now let $\Phi$ be a monotone graph property. Then $\epk\setminus \{\emptyset \}$ is a simplicial complex, called the \emph{graph complex} of $\Phi_k$. The ground set is the set of all edges of the complete labeled graph $K_k$ on $k$ vertices and we emphasize $\epk\setminus \{\emptyset \}$ being a simplicial complex for monotone properties by denoting it as $\Delta(\Phi_k)$. If $\Gamma$ is any permutation group on the set $[k]$ then $\Gamma$ induces a group operation on the ground set of $\Phi_k$, i.e. the edges of the labeled complete graph of size $k$, by relabeling the vertices according to the group element. In particular, $\Delta(\Phi_k)$ is a $\Gamma$-simplicial complex as $\Phi_k$ is invariant under relabeling of vertices. We write~$\Delta^\Gamma(\Phi_k)$ for the fixed-point complex $\Delta(\Phi_k)^\Gamma$ under this operation.

\paragraph*{Parameterized (counting) complexity}
We will follow the definitions of Chapt.~14 of the textbook of Flum and Grohe~\cite{flumgrohe}. A \emph{parameterized counting problem} is a function $F:\{0,1\}^\ast \rightarrow \mathbb{N}$ together with a computable parameterization $\kappa: \{0,1\}^\ast \rightarrow \mathbb{N}$. $(F,\kappa)$ is called \emph{fixed-parameter tractable} (FPT) if there exists a deterministic algorithm $\mathbb{A}$ and a computable function $f$ such that $\mathbb{A}$ computes $F$ in time $f(\kappa(x))\cdot |x|^{O(1)}$ for any input $x$. Given two parameterized counting problems $(F,\kappa)$ and $(F',\kappa')$, a \emph{parameterized Turing reduction} from $(F,\kappa)$ to $(F',\kappa')$ is an FPT algorithm w.r.t. $\kappa$ that has oracle access to $F'$ and that on input $x$ computes $F(x)$ with the additional restriction that there exists a computable function $g$ such that for any oracle query $y$ it holds that $\kappa'(y) \leq g(\kappa(x))$. We write $(F,\kappa) \parared (F',\kappa')$.

The parameterized counting problem $\#\mathsf{Clique}$ aks, given a graph $G$ and a natural number $k$, to compute the number of complete subgraphs of size $k$ in $G$ and the problem is parameterized by $k$. The class $\#\W$ contains all problems $(F,\kappa)$ such that $(F,\kappa) \parared \#\mathsf{Clique}$ holds. Given a recursively enumerable class of graphs $\mathcal{H}$ the problems $\#\homs(\mathcal{H})$, $\#\embs(\mathcal{H})$, $\#\subs(\mathcal{H})$, $\#\strembs(\mathcal{H})$ and $\#\indsubs(\mathcal{H})$ ask, given a graph $H \in \mathcal{H}$ and an arbitrary (unlabeled) graph $G$, to compute $\#\homs(H,G)$, $\#\embs(H,G)$, $\#\subs(H,G)$, $\#\strembs(H,G)$ and $\#\indsubs(H,G)$, respectively. All problems are parameterized by $|H|$. As stated in the introduction, there are dichotomy results for each of the aforementioned problems \cite{homsdicho1,homsdicho2,embsdicho,strembsdicho}. We emphasize on the following, which is crucial for the framework of graph motif parameters.
\begin{theorem}[\cite{homsdicho1,homsdicho2}]\label{thm:homsdicho}
The problem $\#\homs(\mathcal{H})$ is fixed-parameter tractable if there exists $b \in \mathbb{N}$ such that the treewidth\footnote{We remark that the graph parameter of treewidth is not used explicitely in this work. Hence we refer the reader e.g. to Chapt.~11 in~\cite{flumgrohe}.} of every graph in $\mathcal{H}$ is bounded by $b$. Otherwise, the problem is $\#\W$-hard.
\end{theorem}

In this work we deal with a generalization of $\#\indsubs(\mathcal{H})$. Let $\Phi$ be a computable graph property. The problem $\#\indsubs(\Phi)$ asks, given a graph $G$ and a number $k \in \mathbb{N}$ to compute $\sum_{H \in \Phi_k} \#\indsubs(H,G)$. The parameter is $k$.

\section{Graph motif parameters}\label{sec:hombasis}
In~\cite{hombasis2017} Curticapean, Dell and Marx generalized the problem $\#\homs(\mathcal{H})$ to linear combinations, called \emph{graph motif parameters}. To this end, let $\mathcal{A}$ be a recursively enumerable set of functions $a: \mathcal{G} \rightarrow \mathbb{Q}$ such that $\mathsf{supp}(a)$ is finite. Then the problem $\#\homs(\mathcal{A})$ asks, given $a \in \mathcal{A}$ and a graph $G$, to compute $\sum_{H \in \mathcal{G}} a(H) \cdot \#\homs(H,G)$. The parameter is the description length of $a$, denoted by $|a|$. Their main result states that computing a linear combination of homomorphisms is as hard as computing all terms with non-zero coefficients:

\begin{theorem}[\cite{hombasis2017}]\label{thm:cdm_main}
There exists a deterministic algorithm that, on input a function $a: \mathcal{G} \rightarrow \mathbb{Q}$ with finite support, a graph $F \in \mathsf{supp}(a)$ and a graph $G$ and given oracle access to the function $G \mapsto  \sum_{H \in \mathcal{G}} a(H) \cdot \#\homs(H,G)$, computes $\#\homs(F,G)$ in time $g(|a|)\cdot \#V(G)^{O(1)}$ and additionally satisfies that the number of vertices of every graph $G'$ for which the oracle is queried is of size bounded by $\max_{H \in \mathsf{supp}(a)} \#V(H) \cdot \#V(G)$.
\end{theorem}

Using this result, Curticapean, Dell and Marx proved that the problem $\#\homs(\mathcal{A})$ is fixed-parameter tractable if there is a constant upper bound on the treewidth of all graphs that occur in the support of a function $a \in \mathcal{A}$, and $\#\W$-hard otherwise. After that they showed that all of the problems $\#\embs(\mathcal{H})$, $\#\strembs(\mathcal{H})$, ... are expressible as linear combinations of homomorphisms, immediately implying the existence of dichotomy results for those problems. However, establishing a concrete criterion for fixed-parameter tractability requires to find out which graphs are contained in the support of a function $a$ when the problem is translated to a linear combination of homomorphisms, and this can be highly non-trivial. 

In what follows, we will establish a concrete criterion for properties $\Phi$ such that the coefficient of $K_k$ is non-zero when the function $G \mapsto \sum_{H \in \Phi_k } \#\indsubs(H,G)$ is translated to a linear combination of homomorphisms. This is motivated by the fact that, in this case, Theorem~\ref{thm:cdm_main} allows us to compute the number $\#\homs(K_k,G)$ which is equal to $k!$ times the number of cliques of size $k$ in $G$. As  $\#\mathsf{Clique}$ can not be solved in time $f(k)\cdot \#V(G)^{o(k)}$ for any computable function $f$ under the Exponential Time Hypothesis~\cite{eth1,eth2}, we will not only obtain $\#\W$-hardness but also a tight lower bound under the lense of fine-grained complexity theory.

\begin{theorem}[Theorem~\ref{thm:clique_coef_int} restated]\label{thm:clique_coef_tech} Let $\Phi$ be a graph property, let $k$ be a non-zero natural number and let $a: \mathcal{G} \rightarrow \mathbb{Q}$ be the function such that for all graphs $G$ the following is true
\[\sum_{H \in \Phi_k} \#\indsubs(H,G) = \sum_H a(H) \cdot \#\homs(H,G) \,. \]
Then $|k! \cdot a(K_k)| = |\sum_{A \in \epk} (-1)^{\#A}|$. 
\end{theorem}
\begin{proof}
Using the principle of inclusion-exclusion we can express the number of strong embeddings as the number of embeddings (see e.g. Chapt.~5.2.3 in~\cite{lovaszbook}):
\begin{equation}\label{eq:mc_1}
\#\strembs(H,G) = \sum_{\substack{H' \supseteq H\\V(H)=V(H')}} (-1)^{\#E(H')-\#E(H)} \cdot \#\embs(H',G) \,,
\end{equation}
where $H'$ ranges over all graphs obtained from $H$ by adding edges. Next we collect terms in (\ref{eq:mc_1}) that correspond to isomorphic graphs. To this end we let $\#\lbrace H' \supseteq H \rbrace$ denote the number of possibilities to add edges to $H$ such that the resulting graph is isomorphic to $H'$. Note that $\#\lbrace K_k \supseteq H \rbrace = 1$ if $H$ has $k$ vertices. We obtain
\begin{equation}\label{eq:mc_2}
\#\strembs(H,G) = \sum_{H' \in \mathcal{G}} (-1)^{\#E(H')-\#E(H)} \cdot \#\lbrace H' \supseteq H \rbrace \cdot \#\embs(H',G) \,.
\end{equation}
Next we translate the number of embeddings to a linear combination of homomorphisms. This can be done using Möbius inversion\footnote{We omit the formal introduction to Möbius inversion as we will only need that $\mu(\emptyset,\emptyset)=1$. We refer the interested reader to~\cite{lovaszbook}, where the concept is introduced and Equation~(\ref{eq:mc_3}) is proved.} (see~\cite{hombasis2017} or Chapt.~5.2.3 in~\cite{lovaszbook}):
\begin{equation}\label{eq:mc_3}
\#\embs(H',G) = \sum_{\rho \geq \emptyset} \mu(\emptyset,\rho)\cdot \#\homs(H'/\rho,G)\,,
\end{equation}
where the sum and the Möbius function $\mu$ are over the partition lattice of $V(H')$ and $H'/\rho$ is obtained from $H'$ by contracting every pair of vertices that is contained in the same block in~$\rho$. We observe that the coefficient of $\#\homs(K_k,G)$ in the above sum is $\mu(\emptyset,\emptyset)=1$ if $H'$ is isomorphic to $K_k$ and zero otherwise as every vertex contraction of a graph with $k$ vertices that is not the complete graph can not result in the complete graph with $k$ vertices. Hence the coefficient of $\#\homs(K_k,G)$ in Equation~(\ref{eq:mc_2}) is precisely $(-1)^{\#E(K_k)-\#E(H)}$. Next we use Fact~\ref{fact:auto} and obtain that
\begin{equation}\label{eq:mc_4}
\sum_{H \in \Phi_k} \#\indsubs(H,G) = \sum_{H \in \Phi_k} \#\strembs(H,G)\cdot \#\auts(H)^{-1}\,.
\end{equation}
It follows that the coefficient $a(K_k)$ of $\#\homs(K_k,G)$ in Equation~(\ref{eq:mc_4}) satisfies
\begin{equation}\label{eq:mc_5}
a(K_k) = \sum_{H \in \Phi_k} (-1)^{\#E(K_k)-\#E(H)}\cdot \#\auts(H)^{-1}\,.
\end{equation}
We now multiply this equation by $k!$, which we interpret as the number $\#\Symgp_k$ of elements of the symmetric group of the $k$ vertices. Taking also the absolute value on both sides allows us to drop the constant factor $(-1)^{\#E(K_k)}$ and we obtain
\begin{equation}\label{eq:mc_6}
|k! \cdot a(K_k)| = \left\lvert\sum_{H \in \Phi_k} (-1)^{\#E(H)}\cdot \frac{\#\Symgp_k}{\#\auts(H)}\right\rvert\,.
\end{equation}
For any graph $H$ in the above sum choose a set $A_0$ of edges of the labeled complete graph~$K_k$ on $k$ vertices such that the corresponding subgraph $G(A_0)$ is isomorphic to $H$. The group $\Symgp_k$ acts on the vertices and thus on the edges of $K_k$ and by the definition of a graph automorphism, the stabilizer of the set $A_0$ has exactly $\#\auts(H)$ elements. On the other hand the orbit of $A_0$ under $\Symgp_k$ is the collection of all sets $A$ such that $G(A) \cong H$. By the Orbit Stabilizer theorem we have $\frac{\#\Symgp_k}{\#\auts(H)} = \# \{A \subseteq E(K_k) ~|~ G(A) \cong H\}$.
Inserting in equation (\ref{eq:mc_6}) we obtain
\begin{equation}\label{eq:mc_8}
|k! \cdot a(K_k)| = \left\lvert\sum_{H \in \Phi_k} \sum_{\substack{A \subseteq E(K_k) \\ G(A) \cong H}} (-1)^{\#E(H)}\right\rvert = \left\lvert\sum_{A \in \epk} (-1)^{\#A}\right\rvert\,.
\end{equation}
\end{proof}
Theorem~\ref{thm:clique_coef_tech} implies the following sufficient criterion for hardness of $\#\indsubs(\Phi)$ which we will use in the remainder of the paper. 

\begin{corollary}\label{cor:criterion}
Let $\Phi$ be a graph property and let $\mathcal{K}=\{k \in \mathbb{N}~|~\sum_{A \in \epk} (-1)^{\#A} \neq 0\}$. If $\mathcal{K}$ is infinite, then $\#\indsubs(\Phi)$ is $\#\W$-hard. If additionally $\mathcal{K}$ is \dense, $\#\indsubs(\Phi)$ can not be solved in time $f(k)\cdot \#V(G)^{o(k)}$ for any computable function $f$, unless ETH fails.
\end{corollary}
\begin{proof}
Theorem~\ref{thm:clique_coef_tech} and Theorem~\ref{thm:cdm_main} induce a parameterized Turing reduction from the problem $\#\homs(\{K_k ~|~k \in \mathcal{K} \})$ which is known to be $\#\W$-hard by Theorem~\ref{thm:homsdicho}. While this implies the first statement, we explicitely use a reduction from the problem $\#\mathsf{ColClique}$ to prove the latter. $\#\mathsf{ColClique}$ asks, given $k \in \mathbb{N}$ and a $k$-vertex-colored graph $G$, to compute the number of cliques of size $k$ in $G$ that are colorful, i.e. exactly one vertex of each color is contained in the clique. It is known that this problem can not be solved in time $f(k)\cdot \#V(G)^{o(k)}$ for any computable function $f$ unless ETH fails (see e.g. Chapt.~1.2.3 in~\cite{thesis_radu}). Before we proceed with the reduction, we recall that $\mathcal{K}$ being dense implies that there are constants $c$ and $b$ such that for all $k \in \mathbb{N}$ with $k>b$ there exists $k' \in \mathcal{K}$ such that $k\leq k' \leq ck$.
 
Now given an instance $(G,k)$ of $\#\mathsf{ColClique}$ we proceed as follows. If $k \leq b$ we solve the problem by brute-force which requires time $O(n^b)$. Otherwise we search for the minimal number $k' \in \mathcal{K}$ such that $k \leq k' \leq ck$. Next we construct the graph $G'$ from $G$ by adding $k'-k$ vertices $v_{k+1},\dots,v_k'$ and color them with new colors $k+1,\dots,k'$. After that we add edges $\{v_i,u\}$ for all $i\in \{k+1,\dots,k'\}$ and $u \in V(G)$. Now it can easily be verified that the number of colorful $k'$-cliques in $G'$ equals the number of colorful $k$-cliques in $G$. Theorem~\ref{thm:clique_coef_tech} implies that for every graph $G$ the coefficient $a(K_{k'})$ of $\#\homs(K_{k'},G)$ is non-zero if $\sum_{H \in \Phi_{k'}} \#\indsubs(H,G)$ is expressed as a linear combination of homomorphisms and Theorem~\ref{thm:cdm_main} hence allows us to compute $\#\homs(K_{k'},G'')$ for every subgraph $G''$ of $G'$ in FPT time if access to $\#\indsubs(\Phi)$ is provided. Dividing by $k'!$ yields the number of (uncolored) $k'$-cliques in $G''$. Finally, we use the principle of inclusion-exclusion to compute the number of colorful $k'$-cliques in $G'$. (see e.g. Chapt.~1.4.1 in~\cite{thesis_radu}). 

As all oracle calls satisfy that the parameter ($k'$) is bounded by $c\cdot k$ for a constant $c$ and that the size of the queried graph is bounded by $g(k)\cdot \#V(G)$ for some computable function $g$, and the overall reduction runs in FPT time, it holds that any algorithm that solves $\#\indsubs(\Phi)$ in time $f'(k)\cdot \#V(G)^{o(k)}$ for some computable function $f'$ can be used to solve $\#\mathsf{ColClique}$ in time $f''(k)\cdot \#V(G)^{o(k)}$ for some computable function $f''$, which is impossible unless ETH fails.
\end{proof}

\section{Monotone properties}\label{sec:mono}
Recall that monotone graph properties are closed under the removal of edges. In what follows we assume every monotone graph property to hold on the independent set, i.e., the graph containing no edges, because otherwise the property would be trivially false. For technical reasons we say that a property is \emph{non-trivial} if it is false on $K_k$ for all but finitely many $k \in \mathbb{N}$.\footnote{This is required to exlude properties like $\Phi(G) = 0 \Leftrightarrow \#V(G) \equiv 1 \mod 2$ which indeed is monotone as it is closed under the removal of edges.} We start by refining Theorem~\ref{thm:clique_coef_tech} for monotone properties.

\begin{lemma}\label{lem:clique_coef_int_mono_tech}
Let $\Phi$ be a monotone graph property and let $k$ be a non-zero natural number. Then it holds that $\sum_{A \in \epk} (-1)^{\#A} = \hat{\chi}(\Delta(\Phi_k))$.
\end{lemma}
\begin{proof}
By Theorem~\ref{thm:clique_coef_tech} it holds that $|k! \cdot a(K_k)| = |\sum_{A \in \epk} (-1)^{\#A}|$. The claim follows as
\begin{align}
\sum_{A \in \epk} (-1)^{\#A} &= (-1)^{\#\emptyset} + \sum_{i\geq 1} (-1)^i \cdot \#\{A \in \epk \setminus \{\emptyset \} ~|~\#A =i \}\label{eq:is_is_true}\\
~&= 1 + \sum_{i\geq 0} (-1)^{i+1} \cdot \#\{A \in \epk \setminus \{\emptyset \} ~|~\#A =i+1 \}\\
~&= 1 - \sum_{i\geq 0} (-1)^i \cdot \#\{A \in \Delta(\Phi_k) ~|~\mathsf{dim}(A) =i \}\\
~&= 1 - \chi(\Delta(\Phi_k)) = \hat{\chi}(\Delta(\Phi_k))\,.
\end{align}
Note that (\ref{eq:is_is_true}) holds because $\Phi_k$ is true for the independent set.
\end{proof}

\begin{corollary}[Corollary~\ref{cor:eul_non_zero_hard_int} restated] \label{cor:criterion_mono}
Let $\Phi$ be a monotone graph property and let \[\mathcal{K}=\{k \in \mathbb{N}~|~\hat{\chi}(\Delta(\Phi_k)) \neq 0\} \,. \] If $\mathcal{K}$ is infinite, then $\#\indsubs(\Phi)$ is $\#\W$-hard. If additionally $\mathcal{K}$ is \dense, $\#\indsubs(\Phi)$ can not be solved in time $f(k)\cdot \#V(G)^{o(k)}$ for any computable function $f$, unless ETH fails.
\end{corollary}
\begin{proof}
Follows immediately from Lemma~\ref{lem:clique_coef_int_mono_tech} and Corollary~\ref{cor:criterion}.
\end{proof}
The above criterion yields hardness of $\#\indsubs(\Phi)$ for every monotone graph property $\Phi$ whose graph complex is well-understood with respect to the (reduced) Euler characteristic. The thesis of Jonsson (see Chapt.~10.5 in~\cite{jonsson}) provides a large list of graph complexes including e.g. disconnectivity, colorability and coverability, only to name a few, whose reduced Euler characteristics are non-zero infinitely often and hence to which Corollary~\ref{cor:criterion_mono} is applicable. We would also like to point out the work of Chakrabarti, Khot and Shi~\cite{chakrabarti_evasiveness} who proved the reduced Euler characteristic of a large family of graph complexes to be odd. Their result will be used to prove the fourth condition of Theorem~\ref{thm:main} and reads as follows --- we state it in terms of homomorphisms.

\begin{lemma}[\cite{chakrabarti_evasiveness}]\label{lem:chakrabarti} Let $F$ be a graph and let $\Phi^{[F]}$ be the graph property that holds true on a graph $G$ if and only if $\homs(F,G) = \emptyset$, i.e., there is no homomorphism from $F$ to $G$. Furthermore let $T_F := \min \{2^{2^t}-1~|~2^{2^t} \geq \#V(F)\}$ and let $k \in \mathbb{N}$ such that $k \equiv 1 \mod T_F$. Then it holds that $\chi(\Phi_k^{[F]}) \equiv 0 \mod 2$ and hence $\hat{\chi}(\Phi_k^{[F]}) \equiv 1 \mod 2$.
\end{lemma} 

Unfortunately, as the proof of the above lemma shows, it is often quite tedious to argue about the (reduced) Euler characteristic of the graph complex induced by a more complicated property $\Phi$ and hence proving hardness of $\#\indsubs(\Phi)$. In the remainder of this section we will therefore demonstrate that Corollary~\ref{cor:criterion_mono} together with Theorem~\ref{thm:smith} yields a fruitful topological approach to prove $\#\W$-hardness and conditional lower bounds for $\#\indsubs(\Phi)$, given that $\Phi$ is a monotone graph property. We outline the approach in the following lemma. 

\begin{lemma}\label{lem:top_approach}
Let $\Phi$ be a monotone graph property, let $\mathcal{K}$ be an infinite subset of $\mathbb{N}$ and let $\Gamma = \{\Gamma_k~|~ k \in \mathcal{K}\}$ be a set of permutation groups such that for every $k \in \mathcal{K}$ the group $\Gamma_k$ is a $p_k$-power group for some prime $p_k$. If for every $k \in \mathcal{K}$ it holds that
\[ \hat{\chi}(\Delta^\Gamma(\Phi_k))  \not\equiv 0 \mod p_k  \,,\]
then $\#\indsubs(\Phi)$ is $\#\W$-hard. If additionally $\mathcal{K}$ is \dense, $\#\indsubs(\Phi)$ can not be solved in time $f(k)\cdot \#V(G)^{o(k)}$ for any computable function $f$, unless ETH fails.
\end{lemma}
\begin{proof}
Follows immediately from Corollary~\ref{cor:criterion_mono} and Theorem~\ref{thm:smith}.
\end{proof} 

Intuitively, Lemma~\ref{lem:top_approach} states that instead of analyzing $\hat{\chi}(\Delta(\Phi_k))$ which might be tedious, it suffices to prove that the reduced Euler characteristic of the fixed-point complex of $\Delta(\Phi_k)$ with respect to a $p$-power group is not $0$ modulo $p$. For our purposes it will suffice to use the groups $\mathbb{Z}_p$ for prime numbers $p$, explained as follows. Recall that the ground set of $\Delta(\Phi_p)$ is the set of all edges of the labeled complete graph on $p$ vertices. Now $b \in \mathbb{Z}_p$ is interpreted as a relabeling $x \mapsto x+b$ of the vertices\footnote{Here $+$ is addition modulo $p$.}, which induces an operation on the edges by mapping the edge $\{x,y\}$ to the edge $\{x+b,y+b\}$. We remark that this group was also used in an intermediate step in~\cite{topological_evasiveness}. It can easily be verified that this mapping is a group operation. Furthermore $\Delta(\Phi_p)$ is a $\mathbb{Z}_p$-simplicial complex with respect to this operation as $\Phi$ is invariant under relabeling of vertices. Hence the fixed-point complex $\Delta^{\mathbb{Z}_p}(\Phi_p)$ is defined. Furthermore observe that every orbit of the group operation is an Hamilton cycle. We illustrate $\Delta^{\mathbb{Z}_7}(\Phi_7)$ for some properties $\Phi$ in Figure~\ref{fig:p_seven}.

\begin{figure}
\begin{tikzpicture}[-,thick, scale=0.3856]
  \node[circle,inner sep=1pt,draw](0) at (90+0*360/7:2) {\small $0$};
  \node[circle,inner sep=1pt,draw](1) at (90+1*360/7:2) {\small $1$};
  \node[circle,inner sep=1pt,draw](2) at (90+2*360/7:2) {\small $2$};
  \node[circle,inner sep=1pt,draw](3) at (90+3*360/7:2) {\small $3$};
  \node[circle,inner sep=1pt,draw](4) at (90+4*360/7:2) {\small $4$};
  \node[circle,inner sep=1pt,draw](5) at (90+5*360/7:2) {\small $5$};
  \node[circle,inner sep=1pt,draw](6) at (90+6*360/7:2) {\small $6$};
  \node(7) at (0,-3.5) {\footnotesize $\mathcal{O}_1$};
  
  \draw (0)--(1);\draw (1)--(2);\draw (2)--(3);\draw (3)--(4);\draw (4)--(5);
  \draw (5)--(6);\draw (6)--(0);
  
 \begin{scope}[xshift=6cm]
 \node[circle,inner sep=1pt,draw](0) at (90+0*360/7:2) {\small $0$};
  \node[circle,inner sep=1pt,draw](1) at (90+1*360/7:2) {\small $1$};
  \node[circle,inner sep=1pt,draw](2) at (90+2*360/7:2) {\small $2$};
  \node[circle,inner sep=1pt,draw](3) at (90+3*360/7:2) {\small $3$};
  \node[circle,inner sep=1pt,draw](4) at (90+4*360/7:2) {\small $4$};
  \node[circle,inner sep=1pt,draw](5) at (90+5*360/7:2) {\small $5$};
  \node[circle,inner sep=1pt,draw](6) at (90+6*360/7:2) {\small $6$};
  \node(7) at (0,-3.5) {\footnotesize $\mathcal{O}_2$};
  \draw (0)--(2);\draw (2)--(4);\draw (4)--(6);\draw (6)--(1);\draw (1)--(3);
  \draw (3)--(5);\draw (5)--(0);
 \end{scope}
 
 \begin{scope}[xshift=12cm]
 \node[circle,inner sep=1pt,draw](0) at (90+0*360/7:2) {\small $0$};
  \node[circle,inner sep=1pt,draw](1) at (90+1*360/7:2) {\small $1$};
  \node[circle,inner sep=1pt,draw](2) at (90+2*360/7:2) {\small $2$};
  \node[circle,inner sep=1pt,draw](3) at (90+3*360/7:2) {\small $3$};
  \node[circle,inner sep=1pt,draw](4) at (90+4*360/7:2) {\small $4$};
  \node[circle,inner sep=1pt,draw](5) at (90+5*360/7:2) {\small $5$};
  \node[circle,inner sep=1pt,draw](6) at (90+6*360/7:2) {\small $6$};
  \node(7) at (0,-3.5) {\footnotesize $\mathcal{O}_3$};
  \draw (0)--(3);\draw (3)--(6);\draw (6)--(2);\draw (2)--(5);\draw (5)--(1);
  \draw (1)--(4);\draw (4)--(0);
 \end{scope}
 
 \begin{scope}[xshift=18cm]
 \node[circle,inner sep=1pt,draw](0) at (90+0*360/7:2) {\small $0$};
  \node[circle,inner sep=1pt,draw](1) at (90+1*360/7:2) {\small $1$};
  \node[circle,inner sep=1pt,draw](2) at (90+2*360/7:2) {\small $2$};
  \node[circle,inner sep=1pt,draw](3) at (90+3*360/7:2) {\small $3$};
  \node[circle,inner sep=1pt,draw](4) at (90+4*360/7:2) {\small $4$};
  \node[circle,inner sep=1pt,draw](5) at (90+5*360/7:2) {\small $5$};
  \node[circle,inner sep=1pt,draw](6) at (90+6*360/7:2) {\small $6$};
  \node(7) at (0,-3.5) {\footnotesize $\mathcal{O}_1 \cup \mathcal{O}_2$};
  \draw (0)--(1);\draw (1)--(2);\draw (2)--(3);\draw (3)--(4);\draw (4)--(5);
  \draw (5)--(6);\draw (6)--(0);
  \draw (0)--(2);\draw (2)--(4);\draw (4)--(6);\draw (6)--(1);\draw (1)--(3);
  \draw (3)--(5);\draw (5)--(0);
 \end{scope}
 
 \begin{scope}[xshift=24cm]
 \node[circle,inner sep=1pt,draw](0) at (90+0*360/7:2) {\small $0$};
  \node[circle,inner sep=1pt,draw](1) at (90+1*360/7:2) {\small $1$};
  \node[circle,inner sep=1pt,draw](2) at (90+2*360/7:2) {\small $2$};
  \node[circle,inner sep=1pt,draw](3) at (90+3*360/7:2) {\small $3$};
  \node[circle,inner sep=1pt,draw](4) at (90+4*360/7:2) {\small $4$};
  \node[circle,inner sep=1pt,draw](5) at (90+5*360/7:2) {\small $5$};
  \node[circle,inner sep=1pt,draw](6) at (90+6*360/7:2) {\small $6$};
  \node(7) at (0,-3.5) {\footnotesize $\mathcal{O}_1 \cup \mathcal{O}_3$};
  \draw (0)--(1);\draw (1)--(2);\draw (2)--(3);\draw (3)--(4);\draw (4)--(5);
  \draw (5)--(6);\draw (6)--(0);
  \draw (0)--(3);\draw (3)--(6);\draw (6)--(2);\draw (2)--(5);\draw (5)--(1);
  \draw (1)--(4);\draw (4)--(0);
 \end{scope}
 
 \begin{scope}[xshift=30cm]
 \node[circle,inner sep=1pt,draw](0) at (90+0*360/7:2) {\small $0$};
  \node[circle,inner sep=1pt,draw](1) at (90+1*360/7:2) {\small $1$};
  \node[circle,inner sep=1pt,draw](2) at (90+2*360/7:2) {\small $2$};
  \node[circle,inner sep=1pt,draw](3) at (90+3*360/7:2) {\small $3$};
  \node[circle,inner sep=1pt,draw](4) at (90+4*360/7:2) {\small $4$};
  \node[circle,inner sep=1pt,draw](5) at (90+5*360/7:2) {\small $5$};
  \node[circle,inner sep=1pt,draw](6) at (90+6*360/7:2) {\small $6$};
  \node(7) at (0,-3.5) {\footnotesize $\mathcal{O}_2 \cup \mathcal{O}_3$};
  \draw (0)--(2);\draw (2)--(4);\draw (4)--(6);\draw (6)--(1);\draw (1)--(3);
  \draw (3)--(5);\draw (5)--(0);
  \draw (0)--(3);\draw (3)--(6);\draw (6)--(2);\draw (2)--(5);\draw (5)--(1);
  \draw (1)--(4);\draw (4)--(0);
 \end{scope}
 
 \begin{scope}[xshift=36cm]
 \node[circle,inner sep=1pt,draw](0) at (90+0*360/7:2) {\small $0$};
  \node[circle,inner sep=1pt,draw](1) at (90+1*360/7:2) {\small $1$};
  \node[circle,inner sep=1pt,draw](2) at (90+2*360/7:2) {\small $2$};
  \node[circle,inner sep=1pt,draw](3) at (90+3*360/7:2) {\small $3$};
  \node[circle,inner sep=1pt,draw](4) at (90+4*360/7:2) {\small $4$};
  \node[circle,inner sep=1pt,draw](5) at (90+5*360/7:2) {\small $5$};
  \node[circle,inner sep=1pt,draw](6) at (90+6*360/7:2) {\small $6$};
  \node (7) at (0,-3.5) {\footnotesize $\mathcal{O}_1 \cup \mathcal{O}_2 \cup \mathcal{O}_3$};
  \draw (0)--(1);\draw (1)--(2);\draw (2)--(3);\draw (3)--(4);\draw (4)--(5);
  \draw (5)--(6);\draw (6)--(0);
  \draw (0)--(2);\draw (2)--(4);\draw (4)--(6);\draw (6)--(1);\draw (1)--(3);
  \draw (3)--(5);\draw (5)--(0);
  \draw (0)--(3);\draw (3)--(6);\draw (6)--(2);\draw (2)--(5);\draw (5)--(1);
  \draw (1)--(4);\draw (4)--(0);
 \end{scope}

\end{tikzpicture}
\caption{\label{fig:p_seven}Non-empty unions of orbits on the operation of $\mathbb{Z}_7$ on the edge set of the labeled graph with $7$ vertices. If $\Phi$ is trivially true then $\Delta^{\mathbb{Z}_7}(\Phi_7)$ contains all of the above subsets of orbits. If $\Phi$ holds only for bipartite graphs then none of the above subsets is contained in $\Delta^{\mathbb{Z}_7}(\Phi_7)$. If $\Phi$ is planarity then $\Delta^{\mathbb{Z}_7}(\Phi_7) = \{ \mathcal{O}_1 , \mathcal{O}_2 , \mathcal{O}_3\}$. More exotically, if $\Phi$ is the property of not being $5$-edge-connected then $\Delta^{\mathbb{Z}_7}(\Phi_7)$ contains every subset of orbits except for $\mathcal{O}_1 \cup \mathcal{O}_2 \cup \mathcal{O}_3$.}
\end{figure}
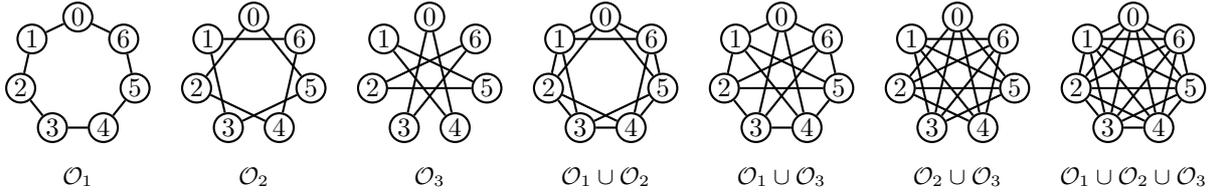

Note that, given a prime $p>2$, the ground set of $\Delta^{\mathbb{Z}_p}(\Phi_p)$ consists of exactly $\frac{1}{2}(p-1)$ elements. In particular those are the Hamiltonian cycles $H_1 = (0,1,2,\dots)$, $H_2=(0,2,4,\dots)$, $H_3=(0,3,6,\dots)$, $\dots$, $H_{\frac{1}{2}(p-1)} =(0,\frac{1}{2}(p-1), p-1,\dots)$. Equivalently, $H_i$ is the orbit of the (labeled) edge $\{0,i\}$ under the operation of $\mathbb{Z}_p$ for $i \in \{1,\dots,\frac{1}{2}(p-1)\}$ and it can easily be verified that those are all orbits of the group operation. In what follows, given a non-empty set $P \subseteq \{1,\dots,\frac{1}{2}(p-1)\}$, we write $H_P$ for the graph with vertices (labeled with) $\{0,\dots,p-1\}$ and edges $\bigcup_{i\in P} H_i$.  

\begin{fact}
Let $P$ be non-empty subset of  $\{1,\dots,\frac{1}{2}(p-1)\}$. Then $P \in \Delta^{\mathbb{Z}_p}(\Phi_p) \Leftrightarrow H_P \in \Phi_p$.
\end{fact}

Now we have everything we need to prove our main result. We start with monotone properties that are false on odd cycles or true on odd antiholes.
\begin{lemma}\label{lem:hardness_odd_cycles_antiholes}
Let $\Phi$ be a non-trivial monotone graph property. If $\Phi$ does not hold on odd cycles or if~$\Phi$ holds on odd anti-holes then there exists a constant $N\in \mathbb{N}$ such that $\hat{\chi}(\Delta^{\mathbb{Z}_p}(\Phi_p)) \not\equiv 0 \mod p$ for every prime $p>N$.
\end{lemma}
\begin{proof}
If $\Phi$ does not hold on odd cycles then $\Delta^{\mathbb{Z}_p}(\Phi_p) = \emptyset$ and hence \[\hat{\chi}(\Delta^{\mathbb{Z}_p}(\Phi_p)) = 1- \chi(\Delta^{\mathbb{Z}_p}(\Phi_p)) = 1-0 = 1 \not\equiv 0 \mod p \,. \]
As $\Phi$ is non-trivial there exists $N \in \mathbb{N}$ such that $\Phi(K_k)=0$ for all $k>N$. Now if $\Phi$ holds on odd anti-holes then $\Delta^{\mathbb{Z}_p}(\Phi_p) = \{ P ~|~ \emptyset \subsetneq P \subsetneq \{1,\dots,\frac{1}{2}(p-1)\} \}$ for all $p> N$ because $\Phi$ is monotone and $H_P$ is an anti-hole if and only if $\#P= \frac{1}{2}(p-1)-1$. Furthermore, $\Phi$ does not hold on $H_{\{1,\dots,\frac{1}{2}(p-1)\}} \cong K_p$ as $p>N$. Hence
\begin{align}
~&\hat{\chi}(\Delta^{\mathbb{Z}_p}(\Phi_p)) \\
~&=  \sum_{i\geq 0} (-1)^i \cdot \# \{ P \in \Delta^{\mathbb{Z}_p}(\Phi_p) \cup \{\emptyset\} ~|~ \#P = i\}\label{eq:use_fact_euler}\\
~&= \sum_{i\geq 0} (-1)^i \cdot \# \{ P \subsetneq \{1,\dots,\frac{1}{2}(p-1)\} ~|~ \#P = i\}\\
~&= \left(\sum_{P \subseteq \{1,\dots,\frac{1}{2}(p-1)\}} (-1)^{\#P}\right) - (-1)^{\frac{1}{2}(p-1)} = (-1)^{\frac{1}{2}(p-1) +1} \not\equiv 0 \mod p \,.
\end{align}
Note that (\ref{eq:use_fact_euler}) follows from Fact~\ref{fact:euler_red_euler}.
\end{proof}

We continue with one more exotic property which illustrates the utility of the topological approach by exploiting the simple structure of $\Delta^{\mathbb{Z}_p}(\Phi_p)$.

\begin{lemma}\label{lem:hardness_connectivite}
Let $c \in \mathbb{N}$ be an arbitrary constant and let $\Phi$ be the graph property of being not $(c+1)$-edge-connected. Then $\hat{\chi}(\Delta^{\mathbb{Z}_p}(\Phi_p)) \not\equiv 0 \mod p$ for every prime $p>c+3$.
\end{lemma}
\begin{proof}
We rely on the following observation.
\begin{claim}\label{clm:connectivity}
The graph $H_P$ is $(c+1)$-edge-connected if and only if $\#P > \lfloor \frac{c}{2} \rfloor$.
\end{claim}
\begin{proof}
If $\#P \leq \lfloor \frac{c}{2} \rfloor$ then every vertex in $H_P$ has degree at most $c$, hence $H_P$ is not $(c+1)$-edge-connected. If $\#P > \lfloor \frac{c}{2} \rfloor$ then $H_P$ contains at least $\lfloor \frac{c}{2} \rfloor + 1$ pairwise edge-disjoint Hamilton cycles. Disconnecting the graph would require to remove at least two edges from every Hamilton cycle, i.e., at least $2 \cdot (\lfloor \frac{c}{2} \rfloor + 1) \geq c+1$ edges. Hence $H_P$ is $(c+1)$-edge-connected. 
\end{proof}
It follows from the Claim that \[\Delta^{\mathbb{Z}_p}(\Phi_p) = \{ P \subseteq \{1,\dots,\frac{1}{2}(p-1) ~|~ P \neq \emptyset \wedge \#P \leq  \lfloor \frac{c}{2} \rfloor \} \,. \] Hence
\begin{align}
~&\hat{\chi}(\Delta^{\mathbb{Z}_p}(\Phi_p)) \\
~&=  \sum_{i\geq 0} (-1)^i \cdot \# \{ P \in \Delta^{\mathbb{Z}_p}(\Phi_p) \cup \{\emptyset\} ~|~ \#P = i\}\label{eq:use_fact_euler_2}\\
~&= \sum_{i=0}^{\lfloor \frac{c}{2} \rfloor } (-1)^i \cdot \# \{ P \subseteq \{1,\dots,\frac{1}{2}(p-1)\} ~|~ \#P= i\}\\
~&= \sum_{i=0}^{\lfloor \frac{c}{2} \rfloor } (-1)^i \cdot \binom{\frac{1}{2}(p-1)}{i} = (-1)^{\lfloor \frac{c}{2} \rfloor} \cdot \binom{\frac{1}{2}(p-1)-1}{\lfloor \frac{c}{2} \rfloor} \not\equiv 0 \mod p \,.
\end{align}
Note that (\ref{eq:use_fact_euler_2}) follows from Fact~\ref{fact:euler_red_euler}.
\end{proof}
Finally, Theorem~\ref{thm:main} follows from Lemma~\ref{lem:top_approach},~\ref{lem:chakrabarti},~\ref{lem:hardness_odd_cycles_antiholes} and~\ref{lem:hardness_connectivite}.

\begin{proof}[Proof of Theorem~\ref{thm:main}]
If $\Phi$ is non-trivial and one of the conditions 1, 2 or 3 is true,  then Lemma~\ref{lem:hardness_odd_cycles_antiholes} and Lemma~\ref{lem:hardness_connectivite} imply the existence of a constant $N \in \mathbb{N}$ such that \[\hat{\chi}(\Delta^{\mathbb{Z}_p}(\Phi_p)) \not\equiv 0 \mod p\] holds for all primes $p>N$. Hence we can apply Lemma~\ref{lem:top_approach} by setting $\mathcal{K}$ to be the set of all primes $p>N$ and $\Gamma_p = \mathbb{Z}_p$ for all $p \in \mathcal{K}$. As $\mathcal{K}$ is obviously infinite, it only remains to show that it is \dense ~as well. However, this is an immediate consequence of Bertrand's postulate, stating that for all natural numbers $n> 3$ there exists at least one prime number $p$ such that $n<p< 2n-2$. 

If $\Phi$ satisfies condition 4, then the claim follows by Lemma~\ref{lem:chakrabarti} and Corollary~\ref{cor:criterion_mono}. To see this, note that the set \[\{k\in \mathbb{N}~|~k\equiv 1 \mod T_F\} \]
is infinite and \dense ~for every fixed graph $F$.
\end{proof}


\section{Non-monotone properties}\label{sec:non_mono}
In this section we turn to non-monotone properties and illustrate that Theorem~\ref{thm:clique_coef_int} itself is a useful criterion when it comes to establishing $\#\W$-hardness of $\#\indsubs(\Phi)$. Recall that, given a prime $q$ and a subset $\mathcal{Q}$ of $\{0,\dots,q-1\}$, the property $\modqa[q,\mathcal{Q}]$ holds on a graph $H$ if and only if $\left(\#E(H)~ \mathsf{mod} ~q\right) \in \mathcal{Q}$. Note that $\modqa[q,\mathcal{Q}]$ generalizes the property of having an even (or odd) number of edges as investigated in \cite{indsubs_eveness}. It turns out that any non-trivial modularity constraint with respect to a prime induces $\#\W$-hardness.

\begin{lemma}\label{lem:non_mono_one}
 Let $q$ be a prime number and $\mathcal{Q} \subseteq \{0,1,\ldots, q-1\}$ a subset which is neither empty nor the full set. Then for $\Phi=\modqa[q,\mathcal{Q}]$ and sufficiently large integers $n$, the sum $\sum_{A \in \epk} (-1)^{\#A}$ is non-vanishing for some $k \in \{n,n+1,n+2\}$.
\end{lemma}
\begin{proof}
 In the case $q=2$ all terms in the sum $\sum_{A \in \epk} (-1)^{\#A}$ have the same sign, so clearly the sum is never zero for $k \geq 1$. Thus we can assume $q \geq 3$.
 
 For $a=0,1,\ldots, q-1$ denote 
 \[S_a(m)=\sum_{j \equiv a ~\mathsf{mod}~ q} (-1)^j \binom{m}{j}.\]
 Given $n \geq 1$ the complete graph on $n$ vertices has $m=\binom{n}{2}$ edges and the number of subgraphs $G$ with $j$ edges is exactly $\binom{m}{j}$. Thus if we define
 \begin{equation} \label{eqn:modulo}
  S_\mathcal{Q}(m) = \sum_{a \in \mathcal{Q}} S_a(m)
 \end{equation}
 then we need to show that for $n$ sufficiently large there is $k \in \{n,n+1,n+2\}$ such that $S_\mathcal{Q}(m) \neq 0$ for $m=\binom{k}{2}$.
 
 The crucial point we are going to use is that the functions $S_a(m)$ satisfy a simple recursion. Indeed
 \begin{align*}S_a(m+1) &= \sum_{j \equiv a ~\mathsf{mod}~ q} (-1)^j \binom{m+1}{j}\\ &= \sum_{j \equiv a ~\mathsf{mod}~ q} (-1)^j \binom{m}{j} + (-1)^j \binom{m}{j-1} \\ &= S_a(m) - S_{a-1}(m), \end{align*}
 where it is understood that $S_{-1}(m)=S_{q-1}(m)$. Let $\vec S = (S_0, S_1, \ldots, S_{q-1})^T$, then expressing the recursion in matrix form we have
 \begin{equation} \label{eqn:matrixrecursion}
  \vec S(m+1) = M \vec S(m),\text{ for } 
  M=\begin{pmatrix}
     1 & 0 & \cdots & 0 & -1\\
     -1 & 1 &  & 0 & 0 \\
     0 & \ddots & \ddots & & \vdots\\
     0 & \cdots & -1&1 & 0\\
     0 & \cdots & 0&-1 & 1
    \end{pmatrix}.
 \end{equation}
In particular, we have
\begin{equation} \label{eqn:matpowerform} \vec S(m) = M^m \vec S(0),\text{ with } \vec S(0) = (1,0, \ldots, 0)^T.\end{equation}
Fortunately, it turns out to be easy to diagonalize the matrix $M$: for the $q$-th root of unity $\omega=\exp(2 \pi i/q)$ we have that the vectors
\[v_b = \frac{1}{q} \begin{pmatrix} 1 \\ \omega^b \\ \vdots \\ \omega^{b (q-1)} \end{pmatrix},\text{ for } b = 0,1, \ldots, q-1\]
are eigenvectors of $M$ for the eigenvalues $\lambda_b = 1-\omega^{-b}$. Moreover, we have that $\vec S(0)=v_0 + v_1 + \ldots + v_{q-1}$ is the sum of all these eigenvectors. Combining $M v_b = \lambda_b v_b$ with equation (\ref{eqn:matpowerform}) we obtain
\begin{equation} \label{eqn:powerfinal}
 \vec S(m) = M^m  \sum_{b=0}^{q-1} v_b=\sum_{b=0}^{q-1} \lambda_b ^m v_b = \frac{1}{q} \sum_{b=0}^{q-1} (1-\omega^{-b})^m \begin{pmatrix} 1 \\ \omega^b \\ \vdots \\ \omega^{b (q-1)} \end{pmatrix}.
\end{equation}
The desired number $S_\mathcal{Q}(m)$ is obtained by summing the components of this vector corresponding to $\mathcal{Q} \subset \{0,1,\ldots, q-1\}$ and we obtain
\begin{equation} \label{eqn:S_A}
 S_\mathcal{Q}(m) = \frac{1}{q} \sum_{b=0}^{q-1} (1-\omega^{-b})^m \left( \sum_{a \in \mathcal{Q}} \omega^{b a}\right).
\end{equation}
Now clearly the term for $b=0$ vanishes. As for the terms with $b \neq 0$ we see then that $\omega^b$ is again a primitive $q$-th root of unity. We claim that $\sum_{a \in \mathcal{Q}} \omega^{b a}$ is then nonzero. Indeed, otherwise the number $\omega^b$ is a zero of the polynomial $P_\mathcal{Q}(z)=\sum_{a \in \mathcal{Q}} z^a$.

In general it is true that a polynomial $P \in \mathbb{Q}[z]$ satisfies $P(\omega^b)=0$ iff $P$ is of the form $P(z)=Q(z) \Phi_q(z)$ with $\Phi_q(z)=z^{q-1}+ z^{q-2} + \ldots + z + 1$ the $q$th cyclotomic polynomial and $Q \in \mathbb{Q}[z]$ any polynomial. For this reason, the polynomial $\Phi_q$ is called the \emph{minimal polynomial} of $\omega^b$. 

Returning to the situation above we see that the polynomial $P_\mathcal{Q}$ is of degree at most $q-1$. Now if $P_\mathcal{Q}(\omega^b)=0$ then we could write $P_\mathcal{Q}(z)=Q(z) \Phi_p(z)$ and by degree considerations we would need to have $Q$ of degree $0$, i.e. a constant. But since we assumed that $\mathcal{Q}$ is not the empty set or the full set $\{0,1, \ldots, q-1\}$, the polynomial $P_\mathcal{Q}$ is not a rational multiple of $\Phi_q$. This gives a contradiction. 

Now we need to find arguments to show $S_\mathcal{Q}(m) \neq 0$ for suitable $m$, which we will later specialize to be of the form $m=\binom{k}{2}$. First, for large $m$ we claim that among the terms $(1-\omega^{-b})^m$ the ones with $b_{\pm}=(q \pm 1)/2$ dominate the others in absolute value. Denote $z_{\pm} = 1-\omega^{-b_{\pm}}$ then we illustrate this in Figure \ref{Fig:rootsofunity} for $q=5$.

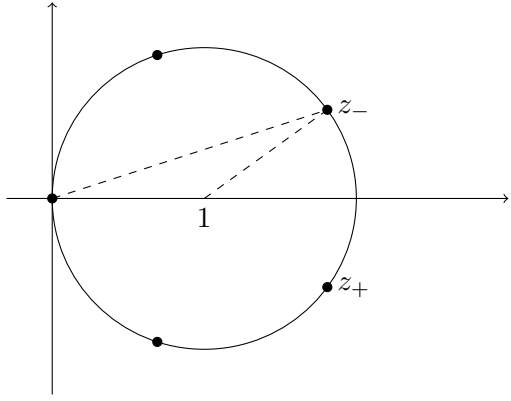
\begin{figure}
 \begin{center}
 \begin{tikzpicture}[scale=2]
  \draw[->] (-0.3,0) -- (3,0);
  \draw[->] (0,-1.3) -- (0,1.3);
  \draw (1,0) circle (1);
  \def\nuPi{3.14159265}
  \foreach \x in {0,1,2,3,4}
    \filldraw[black] ({1-cos(-72*\x)},{sin(72*\x)}) circle (0.03);
  \draw[dashed] (1,0) node[below]{$1$} -- ({1-cos(-72*2)},{sin(72*2)}) node[right]{$z_-$} -- (0,0); 
  \draw ({1-cos(-72*3)},{sin(72*3)}) node[right]{$z_+$}; 
 \end{tikzpicture}
  \end{center}
 \caption{The points $1-\omega^{-b}$ for $q=5$ \label{Fig:rootsofunity}}
\end{figure}

To continue, from the isosceles triangle in Figure \ref{Fig:rootsofunity} one checks that $z_{\pm}$ are positive multiples of primitive $4q$-th roots of unity. Indeed, since $\omega^{-b_-} = \exp(- 2 \pi (q-1)/2q)$, the angle of this triangle at the point $1$ is $2 \pi (q-1)/2q$ and therefore the two identical smaller angles in the triangle are
\[\frac{1}{2} \left(\pi - 2 \pi \frac{q-1}{2q} \right) = \frac{1}{2} \pi \left(1-\frac{q-1}{q} \right) = \frac{1}{2q} \pi = \frac{2 \pi}{4 q}.\]
On the other hand, since they are complex conjugate, one sees that combining the two terms for $b=b_{\pm}$ in the sum (\ref{eqn:S_A}) we obtain
\[\frac{1}{q} \left(z_+^m \left( \sum_{a \in \mathcal{Q}} \omega^{b_+ a}\right) + z_-^m \left( \sum_{a \in \mathcal{Q}} \omega^{b_- a}\right) \right) = \frac{2}{q} \operatorname{Re} z_-^m \left( \sum_{a \in \mathcal{Q}} \omega^{b_- a}\right).\]
Now observe that the absolute value of the term $t_m=z_-^m ( \sum_{a \in \mathcal{Q}} \omega^{b_- a})$ diverges exponentially to infinity and dominates all terms for $b \neq b_{\pm}$ in the sum (\ref{eqn:S_A}). Moreover, the argument of $t_m$ (as a complex number) takes exactly $4q$ different values $\theta_1, \ldots, \theta_{4 q}$, using that $z_-$ is a multiple of a $4q$-th root of unity.
Then we claim that for all $m$ such that $t_m$ is not pure imaginary, its real part still dominates the other summands in (\ref{eqn:S_A}). 
Indeed, we have the explicit estimate
\[|\operatorname{Re} t_m| \geq c |t_m|,\text{ for }c = \min(|\cos(\theta_i)| : i=1, \ldots, 4q\text{ with }\cos(\theta_i) \neq 0 ).\]
But recall that $|t_m|$ diverges exponentially with a base $|z_{\pm}|$. The finitely many other terms in the sum (\ref{eqn:S_A}) also have absolute value that is exponential in $m$ but with strictly smaller base. Thus for $m$ large, the term $c |t_m|$ dominates the combination of all the other summands.

As a conclusion from the claim we must show that for $n$ large, the term $t_m$ is not pure imaginary for $m=\binom{k}{2}$ for some $k \in \{n,n+1,n+2\}$. But assume that $t_{\binom{n}{2}}$ is imaginary, i.e. $\operatorname{arg}(t_{\binom{n}{2}})=\pm \pi/2$. Then since $\binom{n+1}{2} - \binom{n}{2} =n$ and thus $t_{\binom{n+1}{2}} = t_{\binom{n}{2}} z_-^n$, we have
\[\operatorname{arg}(t_{\binom{n+1}{2}}) = \operatorname{arg}(t_{\binom{n}{2}}) + n \operatorname{arg}(z_-) = \pm \frac{\pi}{2} + n \frac{\pi}{2 q}. \]
Now for most $n$ this will already no longer be of the form $\pi/2 + \ell \pi$ and so $t_{\binom{n+1}{2}}$ is not imaginary. In the unlucky case that $2q | n$, we see by the same procedure that then $t_{\binom{n+2}{2}}$ is not imaginary. In any case, we have found a suitable $k$.
\end{proof}

For the second non-monotone property, let $F$ be a connected graph. Then the property $\isof[F]$ holds on a graph $H$ if and only if $H$ contains an isolated subgraph that is isomorphic to $F$.


\begin{lemma}\label{lem:non_mono_two}
 Let $F$ be a connected (unlabeled) graph on $f$ vertices. Then for $\Phi=\isof[F]$ the sum $S_k=\sum_{A \in \epk} (-1)^{\#A}$ is non-vanishing exactly for $k\geq f$ and $k \equiv 0,1$ mod $f$. 
\end{lemma}

\begin{proof}
 Fix $k\geq f$ and let $\mathcal{F}$ be the set of subgraphs of the complete graph $K_k$ isomorphic to $F$. For $F_i \in \mathcal{F}$ let $A_{F_i}$ be the set of graphs on the vertices $[k]$ containing $F_i$ as an isolated subgraph, i.e. a connected component. Then we are interested in the sum
 \[S_k = \sum_{G \in \bigcup A_{F_i}} (-1)^{\#E(G)}.\]
 We compute it via inclusion-exclusion to be
 \begin{equation} \label{eqn:Snisolated} S_k = \sum_{l \geq 1} (-1)^{l+1} \sum_{F_{i_1}, \ldots, F_{i_l} \in \mathcal{F}} \sum_{G \in A_{F_{i_1}} \cap \ldots \cap A_{F_{i_l}}} (-1)^{\#E(G)}\end{equation}
 Note that in the above sum, the graphs $F_{i_1}, \ldots, F_{i_l}$ are assumed pairwise distinct elements of $\mathcal{F}$.
 
 For any $l \geq 1$, we see that the intersection $A_{F_{i_1}} \cap \ldots \cap A_{F_{i_l}}$ is empty if the graphs $F_{i_j}$ are not pairwise vertex-disjoint. Indeed, if two of them share the same vertex there can be no graph containing both of them as isolated subgraphs. On the other hand, if all $F_{i_j}$ are vertex-disjoint, the intersection $A_{F_{i_1}} \cap \ldots \cap A_{F_{i_l}}$ is just the set of graphs containing all of the $F_{i_j}$ as isolated subgraphs (here we use $F$ connected). We can understand this set very explicitly: $F$ has $f$ vertices, so there are $k-lf$ vertices not contained in any $F_{i_j}$ and between those we have full freedom to put edges or not. The total number of possibilities is $2^{\binom{k-lf}{2}}$. Moreover, we can explicitly calculate the sum appearing above as
 \[\sum_{G \in A_{F_{i_1}} \cap \ldots \cap A_{F_{i_l}}} (-1)^{\#E(G)} = \begin{cases}
           (-1)^{l \cdot \#E(F)},&\text{ for }k-lh=0,1\\
           0,&\text{ for }k-lh \geq 2.                                                                         \end{cases}
\]
This is because for $k-lh \geq 2$ and two vertices $v,w \in [k] \setminus \bigcup_{j=1}^l V(F_{i_j})$ the operation of flipping the edge ${v,w}$ gives a bijective map from $A_{F_{i_1}} \cap \ldots \cap A_{F_{i_l}}$ to itself flipping the parity of the number of edges. 

Going back to (\ref{eqn:Snisolated}) let us first treat the special case $f=1$, i.e. $F$ is an isolated vertex. Then all terms for $1 \leq l \leq k-2$ vanish and we are left with
\[S_k = (-1)^{k-1+1} k + (-1)^{k+1} = (-1)^k (k-1)\text{ for }k \geq 2, S_1 = 1.\]
This never vanishes, proving the theorem.

Now assume $f \geq 2$. Then we see that all summands for $1 \leq l < \lfloor k/f \rfloor=l_0$ vanish. Writing $k=l_0 f + a$ with $0 \leq a \leq f-1$ we note that for $a \geq 2$ the remaining summands for $l=l_0$ also vanish. On the other hand, for $a=0,1$ the summands for $l=l_0$ all have the same sign $(-1)^{l+1} (-1)^{l \cdot \#E(F)}$ and there is at least one nonzero summand like this. Thus the sum does not vanish for $k$ of this form. 

In fact we can compute $S_k$ to be 
\[S_k = (-1)^{l+1} (-1)^{l \cdot \#E(F)} \frac{k!}{(f!)^{l_0} l_0!} \left(\frac{f!}{\#\auts(F)} \right)^{l_0} \]
in the cases $a=0,1$. Indeed, the factor $\frac{k!}{(f!)^{l_0} l_0!}$ describes the number of possibilities to choose an unordered collection of $l_0$ sets of size $f$ among the $k$ vertices. For each of these sets there are $\frac{f!}{\#\auts(F)}$ possibilities to put a graph isomorphic to $F$ on the vertices of this set, by the Orbit Stabilizer theorem.
\end{proof}

Now Lemma~\ref{lem:non_mono_one} and \ref{lem:non_mono_two} tell us that for the properties $\Phi=\modqa[q,\mathcal{Q}], \isof[F]$ the set of $k$ such that $\sum_{A \in \epk} (-1)^{\#A} \neq 0$ is dense. Hence Theorem~\ref{thm:main_non_mono} follows by Corollary~\ref{cor:criterion}.

\section{Conclusion and future work}\label{sec:end}
We used the framework of graph motif parameters to provide a sufficient criterion for $\#\W$-hardness of $\#\indsubs(\Phi)$. For monotone properties $\Phi$ this amounts to the reduced Euler\linebreak characteristic of the associated graph complex to be non-zero infinitely often. In particular, our results provide a fine-grained reduction from the problem of counting cliques of size~$k$ to counting induced subgraphs of size $k$ with property $\Phi$ whenever $\Phi$ is monotone and $\hat{\chi}(\Delta(\Phi_k)) \neq 0$. Using a topological approach, we established hardness for a large class of non-trivial monotone graph properties. The obvious next question, whose answer would settle the parameterized complexity of $\#\indsubs(\Phi)$ for monotone properties completely, is whether for every non-trivial monotone property $\Phi$ the set of $k$ such that $\hat{\chi}(\Delta(\Phi_k)) \neq 0$ is infinite. 

\paragraph*{Acknowledgements} We are very grateful to Cornelius Brand, Radu Curticapean, Holger Dell and Johannes Lengler for helpful advice and fruitful discussions.

\bibliographystyle{plainurl}
\bibliography{references}{}

\end{document}